%% file: main.tex
\documentclass[english]{article}
\usepackage[T1]{fontenc}
\usepackage[latin9]{inputenc}
\usepackage{amsmath}
\usepackage{amsthm}
\usepackage{amssymb}
\usepackage{color, xcolor}
\usepackage{graphicx}
\usepackage{caption}
\usepackage{subcaption}
\usepackage{authblk}

\makeatletter
\theoremstyle{plain}
\newtheorem{thm}{\protect\theoremname}
\theoremstyle{plain}
\newtheorem{lem}[thm]{\protect\lemmaname}
\theoremstyle{plain}
\newtheorem{cor}[thm]{\protect\corollaryname}
\theoremstyle{plain}
\newtheorem{clm}[thm]{\protect\claimname}

\usepackage{fullpage}

\makeatother

\usepackage{babel}

\usepackage[textsize=tiny,textwidth=2cm,color=green!50!gray]{todonotes} 

\providecommand{\corollaryname}{Corollary}
\providecommand{\lemmaname}{Lemma}
\providecommand{\theoremname}{Theorem}
\providecommand{\claimname}{Claim}

\begin{document}
\global\long\def\R{\mathcal{R}}%
\global\long\def\N{\mathbb{N}}%
\global\long\def\L{\mathcal{L}}%
\global\long\def\OPT{\mathrm{OPT}}%
\global\long\def\SOL{\mathrm{SOL}}%
\global\long\def\DP{\mathrm{DP}}%
\newcommand{\Exp}{\mathop{\mathbb{E}}}
\newcommand{\horstab}{{\sc Stabbing}}
\newcommand{\orthstab}{{\sc HV-Stabbing}}
\newcommand{\sqstab}{{\sc Square-Stabbing}}
\newcommand{\eps}{\varepsilon}

\title{A PTAS for the horizontal rectangle stabbing problem}

\author[1]{Arindam Khan}
\author[2]{Aditya Subramanian}
\author[3]{Andreas Wiese}

\affil[1]{Indian Institute of Science, Bengaluru, India.
        \texttt{arindamkhan@iisc.ac.in}}
\affil[2]{Indian Institute of Science, Bengaluru, India.
        \texttt{adityasubram@iisc.ac.in}}
\affil[3]{Universidad de Chile, Chile.
        \texttt{awiese@dii.uchile.cl}}



\newcommand{\awr}[1]{}
\newcommand{\adir}[1]{}
\newcommand{\arir}[1]{}

\newcommand{\note}[1]{{#1}}
\newcommand{\aw}[1]{{#1}}
\newcommand{\adi}[1]{{#1}}
\newcommand{\ari}[1]{{#1}}

\date{}

\maketitle

\abstract{We study rectangle stabbing problems \aw{in which} we are given $n$
    axis-aligned rectangles in the plane that we \aw{want to \emph{stab}, i.e., we want to select
    line segments such that for each given rectangle there is a line segment that intersects
    two opposite edges of it.}
    In the {\em horizontal rectangle stabbing problem} (\horstab), the goal is
    to find a set of horizontal line segments of minimum total length such that
    all rectangles are stabbed. In {\em general rectangle stabbing problem},
    also known as {\em horizontal-vertical stabbing problem}
    (\orthstab)\awr{maybe omit the term ``general rectangle stabbing''?}, the
    goal is to find a set of \aw{rectilinear} (i.e., either vertical or
    horizontal) line segments of minimum total length such that all rectangles
    are stabbed.  Both variants are NP-hard.  Chan, van Dijk, Fleszar,
    Spoerhase, and Wolff~\cite{ChanD0SW18} initiated the study of these
    problems by providing constant approximation algorithms.  Recently,
    Eisenbrand, Gallato, Svensson, and Venzin~\cite{QPTAS} have presented a
    QPTAS and a polynomial-time 8-approximation algorithm for \horstab~but it
    is was open whether the problem admits a PTAS.

    In this paper, we obtain a PTAS for \horstab, settling this question.  For
    \orthstab, we obtain a $(2+\eps)$-approximation. We also obtain PTASes for
    special cases of \orthstab: (i) when all rectangles are squares, (ii) when
    each rectangle's width is at most its height, and (iii) when all rectangles
    are $\delta$-large, \aw{i.e., have at least one edge whose length is at
    least $\delta$, while all edge lengths are at most 1.} Our result also
    implies improved approximations for other problems such as {\em generalized
    minimum Manhattan network}.
}

\input{introduction}

\section{Dynamic program} \label{sec:DP}

We present a dynamic program that computes a $(1+\eps)$-approximation to
\orthstab~for the case where $h_{i}\ge w_{i}$ for each rectangle $R_{i}\in\R$.
\aw{This implies directly a PTAS for the setting of squares for the same
    problem, and we will argue that it also yields a PTAS for \horstab.  Also,
we will use it later as a subroutine to obtain a $(2+\eps)$-approximation
for~\orthstab~and a PTAS for the setting of $\delta$-large rectangles of
\orthstab.  }

For a \aw{line} segment~$\ell$, we use the notation $|\ell|$ to represent its
length, and for a set of segments $\L$, we use notation $c(\L)$ to represent
the cost of the set, which is also the total length of the segments contained
in it. \adi{We use the term $\OPT$ interchangeably to refer to the optimal
solution to the problem and also to $c(\OPT)$, i.e., the cost of the optimal
solution.}

\subsection{Preprocessing step}
First, we show that by some simple scaling and discretization steps we can
ensure some simple properties that we will use later. Without loss of
generality we assume that  $(1/\eps)\in\N$ and we say that a value
$x\in\mathbb{R}$ is \emph{discretized} if $x$ is an integral multiple of
${\eps}/{n}$.
\begin{lem} \label{lem:preprocess}
    For any positive constant $\eps<1/3$,
    By losing a factor $(1+O(\eps))$ in the approximation ratio, we can assume
    for each $R_{i}\in\R$ the following properties hold:
    \begin{enumerate}
        \item[(i)] ${\eps}/{n}\le w_{i}\le 1$,
        \item[(ii)] $x_{1}^{(i)},x_{2}^{(i)}$ are discretized and within $[0, n]$,
        \item[(iii)] $y_{1}^{(i)},y_{2}^{(i)}$ are discretized and within
            $[0,4n^2]$, and
        \item[(iv)] each horizontal line segment in the optimal solution has
            width of at most $1/\eps$.
    \end{enumerate}
\end{lem}
\begin{proof}
    Let us start with a {\em scaling} step.  We can scale each rectangle in the
    input in both directions such that the largest width of a rectangle is
    $1-2\eps$, and hence $c(\OPT)\ge1-2\eps$.  Now, we see that for all
    rectangles with width at most $\eps/n$ we can stab them greedily using
    segments of total length at most $\eps$.  So, for $\eps<1/3$, we get a
    multiplicative loss of
    \[ \frac{1-\eps}{1-2\eps}=1+\eps\cdot\frac{1}{1-2\eps}
    <1+\eps\cdot\frac{1}{1-2(1/3)}=1+3\cdot\eps.\]
    Hence, with a factor $ (1+O(\eps))$ loss we can assume that for all
    rectangles ${\eps}/{n}\le w_i\le1-2\eps$.

    Now we incorporate a {\em discretization} step. Since each rectangle
    $R_i\in\R$ has width $w_i\le 1-2\eps$, the instance either has some
    $x$-coordinate that is not covered by any rectangle -- in which case we can
    split the instance into smaller subproblems around this $x$-coordinate --
    or the total width of the instance is less than $n(1-2\eps)$.  Hence all
    the $x$-coordinates of the rectangles can be assumed to be between 0 and
    $n(1-2\eps)$. Further, we extend the rectangles on both sides to make their
    $x$-coordinates align with the next nearest multiple of ${\eps}/{n}$.  This
    proves Property $(ii)$. Also this involves extension by at most
    ${2\eps}/{n}$ to the width of every rectangle, and at most a $2\eps$
    addition to the cost of the solution.  Thus, ${\eps}/{n}\le w_i\le1$,
    proving Property $(i)$.

    Next, we show Property $(iii)$ using a {\em stretching} step. Since we
    have $h_i \ge w_i$ for the input instance, the heights of the rectangles
    could still be arbitrarily large.  But after scaling the rectangles, as
    $w_i\le1$ we know that $c(\OPT)$ is clearly less than $n$.  So, there
    cannot be any vertical line segment in $\OPT$ which is longer than $n$.
    We exploit this fact to scale the heights of the rectangles as follows:
    of the at most $2n$ distinct $y$-coordinates of the rectangles, if the
    distance between any consecutive coordinates is more than $2n$, we
    stretch it down to \adi{the nearest multiple of $\eps/n$ which is less
    than} $2n$. This operation does not affect the size of any optimal
    solution since given any such solution, there cannot be a line that spans
    one of these stretched vertical sections (as it would add cost $2n$).
    Hence, any valid optimal solution to this stretched instance is also a
    valid optimal solution to the original instance.  Since all the
    $y$-coordinates are separated by at most $2n$, all the coordinates can be
    assumed to be within  $[0, 4n^2]$, and similar to the $x$-coordinates, we
    have extended rectangles by at most $\eps/n$ on each side with at most
    $2\eps$ addition to the cost of the solution.

    Finally, we consider Property $(iv)$.  Let $\OPT'$ be the optimal solution
    after the above three steps. Then $\OPT'$ only increases the original
    $\OPT$ by  $(1+O(\eps))$-factor.  Consider any horizontal segment
    $\ell\in\OPT'$, that is longer than $1/\eps$. From the left end point, we
    divide the segment into consecutive smaller segments of length $1/\eps-2$
    each, with one potential last piece being smaller than $1/\eps-2$. Now, for
    each smaller segment we extend it on both sides in such a way that it
    completely stabs the rectangles that it intersects.  Since the maximum
    width of a rectangle is 1, we extend each such segment by at most 2 units.
    In the worst case, the highest possible fractional increase due to the
    above division happens when we have a segment of length $1 /\eps+\delta$
    for very small $\delta>0$. But even in this case the fractional increase in
    the length of segments in $\OPT'$ can be bounded by,
    \[\frac{1/\eps+\delta+4}{1/\eps+\delta} =1+4\cdot\frac{1}{1/\eps+\delta}
    \le 1+4\eps.  \]
    This completes the proof of the lemma.
\end{proof}

\begin{figure}
    \centering
    \begin{subfigure}[h]{0.4\textwidth}
        \centering
        \includegraphics[width=\textwidth,page=1]{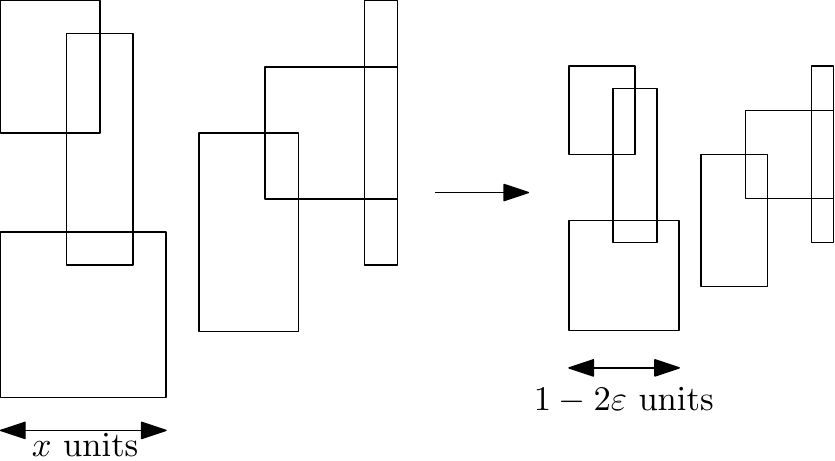}
        \caption{Scaling}
    \end{subfigure}
    \hspace{1.5cm}
    \begin{subfigure}[h]{0.4\textwidth}
        \centering
        \includegraphics[width=\textwidth,page=2]{01preprocess}
        \caption{Discretization}
    \end{subfigure}

    \begin{subfigure}[h]{0.4\textwidth}
        \centering
        \includegraphics[width=\textwidth,page=3]{01preprocess}
        \caption{Stretching}
    \end{subfigure}
    \hspace{1.5cm}
    \begin{subfigure}[h]{0.4\textwidth}
        \centering
        \includegraphics[width=\textwidth,page=4]{01preprocess}
        \caption{Resizing segments in $\OPT$}
    \end{subfigure}
    \caption{Pre-processing steps.}
    \label{fig:preprocess}
\end{figure}

\noindent Henceforth in this paper when we refer to the set of input rectangles
$\R$, we are referring to a set $\R'$ that has been obtained after applying the
pre-processing from Lemma~\ref{lem:preprocess} to the input set $\R$, and when
we refer to $\OPT$, we are referring to the optimal solution to the set of
rectangles $\R'$, which is an $1+O(\eps)$ approximation of the optimal solution
of the input instance.

\subsection{Description of the dynamic program}
\label{sec:DPdesc}
Our algorithm is based on a dynamic program. It
has a cell $\DP(S,\L)$ for each combination of
\begin{itemize}
    \item a rectangle $S\subseteq[0,n]\times[0,4n^2]$ with discretized
        coordinates (that is not necessarily equal to an input rectangle in
        $\R$).
    \item a set $\L$ of at most $3(1/\eps)^{3}$ line segments, each of them
        horizontal or vertical, such that for each $\ell\in\L$ we have that
        $\ell\subseteq S$, and all coordinates of $\ell$ are discretized.
\end{itemize}
This DP-cell corresponds to the subproblem of stabbing all rectangles in $\R$
that are contained in $S$ and that are not already stabbed by the line segments
in $\L$. Therefore, the DP stores solution $\SOL(S,\L)$ in the cell $\DP(S,\L)$
such that $\SOL(S,\L)\cup\L$ stabs all rectangles in $\R$ that are contained in
$S$.

Given a DP-cell $\DP(S,\L)$, our dynamic program computes a solution for it as
follows. If $\L$ already stabs each rectangle from $\R$ that is contained in
$S$, then we simply define a solution $\SOL(S,\L):=\emptyset$ for the cell
$\DP(S,\L)$ and do not compute anything further. Another simple case is when
there is a line segment $\ell\in\L$ such that $S\setminus\ell$ has two
connected components $S_{1},S_{2}$. In this case we define
$\SOL(S,\L):=\SOL(S_{1},\L\cap S_{1})\cup \SOL(S_{2},\L\cap
S_{2})\cup\{\ell\}$, where for any set of line segments $\L'$ and any rectangle
$S'$ we define $\L'\cap S':=\{\ell'\cap S'|\ell'\in\L'\wedge\ell'\cap
S'\ne\emptyset\}$.  In case that there is more than one such line segment
$\ell\in\L$ then we pick one according to some arbitrary but fixed global
tie-breaking rule. We will later refer to this as \emph{trivial operation}.

Otherwise, we do each of the following operations which produces a set of
candidate solutions:
\begin{enumerate}
    \item \emph{Add operation:} Consider each set $\L'$ of line segments with
        discretized coordinates such that $|\L|\cup|\L'|\le3\eps^{-3}$ and
        each $\ell\in\L'$ is contained in $S$ and horizontal or vertical. For
        each such set $\L'$ we define the solution $\L'\cup \SOL(S,\L\cup\L')$
        as a candidate solution.
    \item \emph{Line operation:} Consider each vertical/horizontal line $\ell$
        with a discretized vertical/horizontal coordinate such that
        $S\setminus\ell$ has two connected components $S_{1}$ and $S_{2}$. Let
        $\R_{\ell}$ denote the rectangles from $\R$ that are contained in $S$
        and that are stabbed by $\ell$. For the line $\ell$ we do the
        following:
        \begin{enumerate}
            \item compute an $O(1)$-approximate solution $\L(\R_{\ell})$ for
                the rectangles in $\R_{\ell}$ using the polynomial time
                algorithm in \cite{ChanD0SW18}.
            \item produce the candidate solution $\L(\R_{\ell})\cup
                \SOL(S_{1},\L\cap S_1)\cup \SOL(S_{2},\L\cap S_2)$.
        \end{enumerate}
\end{enumerate}
Note that in the line operation we consider entire lines, not just line
segments. We define $\SOL(S,\L)$ to be the solution of minimum cost among all
the candidate solutions produced above and store it in $\DP(S,\L)$.

We do the operation above for each DP-cell $\DP(S,\L)$. Finally, we output the
solution $\SOL([0,n]\times[0,4n^2],\emptyset)$, i.e., the solution
corresponding to the cell $\DP([0,n]\times[0,4n^2],\emptyset)$.

We remark that instead of using the $O(1)$-approximation algorithm in
\cite{ChanD0SW18} for stabbing the rectangles in $\R_{\ell}$, one could design
an algorithm with a better approximation guarantee, using the fact that all
rectangles in $\R_{\ell}$ are stabbed by the line $\ell$.  However, for our
purposes an $O(1)$-approximate solution is good enough.

\subsection{Definition of DP-decision tree}

We want to show that the DP above computes a $(1+\eps)$-approximate solution.
For this, we define a tree $T$ in which each node corresponds to a cell
$\DP(S,\L)$ of the DP and a corresponding solution $\overline{\SOL}(S,\L)$ to
this cell. The root node of $T$ corresponds to the cell
$\DP([0,n]\times[0,4n^2],\emptyset)$.  Intuitively, this tree represents doing
one of the possible operations above, of the DP in the root problem
$\DP([0,n]\times[0,4n^2],\emptyset)$ and recursively one of the possible
operation in each resulting DP-cell.  The corresponding solutions in the nodes
are the solutions obtained by choosing exactly these operations in each
DP-cell. Since the DP always picks the solution of minimum total cost this
implies that the computed solution has a cost that is at most the cost of the
root, $c(\overline{\SOL}([0,n]\times[0,4n^2],\emptyset))$.

Formally, we require $T$ to satisfy the following properties. We require that a
node $v$ is a leaf if and only if for the corresponding DP-cell $\DP(S,\L)$ the
DP directly defined that $\DP(S,\L)=\emptyset$ because all rectangles in $\R$
that are contained in $S$ are already stabbed by the segments in $\L$. If a
node $v$ for a DP-cell $\DP(S,\L)$ has one child then we require that we reduce
the problem for $\DP(S,\L)$ to the child by applying the add operation, i.e.,
there is a set $\L'$ of horizontal/vertical line segments with discretized
coordinates such that $|\L|\cup|\L'|\le3(1/\eps)^3$, the child node of $v$
corresponds to the cell $\DP(S,\L\cup\L')$, and
$\overline{\SOL}(S,\L)=\overline{\SOL}(S,\L\cup\L')\cup\L'$.

Similarly, if a node $v$ has two children then we require that we can reduce
the problem of $\DP(S,\L)$ to these two children by applying the trivial
operation or the line operation. Formally, assume that the child nodes
correspond to the subproblems $\DP(S_1, \L_1)$ and $\DP(S_2, \L_2)$. If there
is a segment $\ell\in\L$ such that $S_1\cup S_2\cup\ell=S$, then the applied
operation was a trivial operation, and it must also be true that
$\L_1\cup\L_2\cup\{\ell\}=\L$ and
$\overline{\SOL}(S,\L)=\overline{\SOL}(S_1,S_1\cap\L)\cup\overline{\SOL}(S_2,
S_2\cap\L)$. If no such segment exists, then the applied operation was a line
operation on a line along the segment $\ell$, such that $S_1\cup S_2\cup
\ell=S$, $\L_1\cup\L_2=\L$, and
$\overline{\SOL}(S,\L)=\overline{\SOL}(S_1,S_1\cap\L)\cup\overline{\SOL}(S_2,
S_2\cap\L)\cup\L(\R_\ell)$; where $\L(\R_\ell)$ is a $O(1)$-approximate
solution for the set of segments stabbing the set of rectangles intersected by
$\ell$.

We call a tree $T$ with these properties a \emph{DP-decision-tree}.

\begin{lem} \label{lem:runtime}
    If there is a DP-decision-tree $T'$ for which $c(\overline{\SOL}(
    [0,n]\times[0,4n^2],\emptyset))\le(1+\eps)\OPT$ then the DP is a $
    (1+\eps)$-approximation algorithm with a running time
    of $(n/\eps)^{O(1/\eps^3)}$.
\end{lem}
\begin{proof}
    We know that the DP always picks a solution that corresponds to the
    DP-decision tree with the minimum cost. Since $T'$ is a valid DP-decision
    tree, the solution picked by the DP is of cost at most
    $c(\overline{\SOL}([0,n]\times[0,4n^2],\emptyset))$.

    Now let us consider the running time of the algorithm. Since a DP problem
    is defined on a discretized rectangular cell, there  are at most
    $\left(\frac{n}{(\eps/n)}\right)\times
    \left(\frac{4n^2}{(\eps/n)}\right)=\frac{4n^5}{\eps^{2}}$ possible corner
    vertices for the rectangle, and hence $\binom{4n^5/\eps^{2}}{2}$, i.e.,
    $O(n^{10}/\eps^{4})$ possible rectangles.

    Similarly, the subproblem definition also includes a set of segments $\L$
    of size at most $3\eps^{-3}$. Since the segments are discretized, there can
    be at most  $\binom{4n^3/\eps}{2}\times\frac{n^2}{\eps}+
    \binom{n^2/\eps}{2}\times\frac{4n^3}{\eps}$, i.e., $O (n^8/\eps^3)$
    possible segments. So there are at most $(n/\eps)^ {O(1/\eps^3)}$ sets of
    size lesser than $3\eps^{-3}$, and at most $(n/\eps)^{O(1/\eps^3)}$ valid
    DP subproblems.

    For each subproblem we have to consider all possible candidate
    solutions/operations and select the minimum. There are at most
    ${n^2}/{\eps} + {4n^3}/{\eps}$, i.e., $O(n^3/\eps)$
    possible line operations, 
    and $(n/\eps)^{O(1/\eps^3)}$ possible add operations (we can charge
    the trivial operations to the corresponding add operation, and do not have
    to account for them here). This brings the total number of operations to be
    of the order of $(n/\eps)^{O(1/\eps^3)}$. Now in each
    line operation we call the $O(1)$-approximation algorithm from
    \cite{ChanD0SW18} which is also polynomial time.  This
    brings the running time of our algorithm to $(n/\eps)^{O
    (1/\eps^3)}$.
\end{proof}
We define now a DP-decision-tree for which
$c(\overline{\SOL}(S,\L))\le(1+\eps)\OPT$.  Assume w.l.o.g. that $1/\eps\in\N$.
We start by defining a hierarchical grid of vertical lines.  Let $a\in\N_{0}$
be a random offset to be defined later. The grid lines have levels. For each
level $j\in\N_{0}$, there is a grid line
$\{a+k\cdot\eps^{j-2}\}\times\mathbb{R}$ for each $k\in\N$. Note that for each
$j\in\N_{0}$ each grid line of level $j$ is also a grid line of level $j+1$.
Also note that any two consecutive lines of some level $j$ are exactly
$\eps^{j-2}$ units apart.

We say that a line segment $\ell\in\OPT$ is \emph{of level} $j$ if the length
of $\ell$ is in $(\eps^j, \eps^{j-1}]$ (Note that we can have vertical segments
which are longer then $1/\eps$, we consider these also to be of level 0). We
say that a \emph{horizontal} line segment of some level $j$ is
\emph{well-aligned }if both its left and its right $x$-coordinates lie on a
grid line of level $j+3$, i.e., if both of its $x$-coordinates are of the form
$a+k\cdot\eps^{j+1}$. We say that a \emph{vertical} line segment of some level
$j$ is \emph{well-aligned }if both its top and bottom $y$-coordinates are
integral multiples of $\eps^{j+1}$.  This would be similar to the segment's end
points lying on an (imaginary) horizontal grid line of level $j+3$.
\begin{lem}
    By losing a factor $1+O(\eps)$, we can assume that each line segment
    $\ell\in\OPT$ is well-aligned.
\end{lem}
\begin{proof}
    A segment $\ell$ in some level $j$ will be
    of length in $(\eps^j, \eps^{j-1}]$. To align it to a grid line of
    level $j+3$ we would need to extend it by at most $\eps^{j+1}$ on each
    side. The new segment $\ell'$ thus obtained is of length
    \[ |\ell'| \le |\ell| + 2\eps^{j+1} \le |\ell| \left(1+\frac{2\eps^{j+1}}{|\ell|}\right)
    <|\ell|\left(1+\frac{2\eps^{j+1}}{\eps^j}\right)=|\ell|\cdot(1+2\eps).\]
     Therefore, the sum of weights over all the segments in $\OPT$ is \[\sum_{\ell\in\OPT}
    |\ell'|\le \sum_{\ell\in\OPT} |\ell|\cdot(1+2\eps) = (1+2\eps)\cdot\OPT.
    \qedhere \]
\end{proof}

We define the tree $T$ by defining recursively one of the possible operations
(trivial operation, add operation, line operation) for each node $v$ of the
tree. After applying an operation, we always add children to the processed node
$v$ that corresponds to the subproblems that we reduce to, i.e., for a node $v$
corresponding to the subproblem $\DP(S, \L)$, if we are applying the trivial
(resp. line) operation along a segment (resp. line) $\ell$, then we  add children
corresponding to the DP subproblems $\DP(S_1, S_1\cap\L)$ and $\DP(S_2,
S_2\cap\L)$, where $S_1$ and $S_2$ are the connected components of $S\backslash
\ell$. Similarly if we apply the add operation on $v$ with the set of segments
$\L'$ then we add the child node corresponding to the subproblem $\DP(S,
\L\cup\L')$.

\paragraph{First level.}
We start with the root $\DP([0,n]\times[0,4n^2],\emptyset)$. We apply the line
operation for each vertical line that corresponds to a (vertical) grid line
of level 0. Consider one of the resulting subproblems $\DP
(S,\emptyset)$.  Suppose that there are more than $\eps^{-3}$ line
segments from $\OPT$ of level 0 inside $S$. We want to partition $S$ into
smaller rectangles, such that within each of these rectangles $S'$ at most $O
(\eps^{-3})$ of these level 0 line segments start or end.  This will make it
easier for us to guess them. To this end, we consider the line segments from
$\OPT$ of level 0 inside $S$, take their endpoints and order these endpoints
non-decreasingly by their $y$-coordinates. Let $p_{1},p_{2},...,p_{k}$ be
these points in this order.  For each $k'\in\N$ with $k'/\eps^3\le k$, we
consider the point $p_{k'/\eps^3}$. Let $\ell'$ be the horizontal line
that contains $p_{k'/\eps^3}$. We apply the line operation~to $\ell'$.

\begin{clm}
    \label{clm:few-intersecting}Let $\DP(S',\emptyset)$ be one of the
    subproblems after applying the operations above. There are at most
    $\eps^{-3}$ line segments $\L'$ (horizontal or vertical) from $\OPT$ of
    level 0 that have an endpoint inside $S'$.
\end{clm}
In each resulting subproblem $\DP(S',\emptyset)$, for each vertical line segment
$\ell\in\OPT$ that crosses $S'$, i.e., such that $S'\setminus\ell$ has two
connected components, we apply the line operation for the line that contains
$\ell$. In each subproblem $\DP(S'',\emptyset)$ obtained after this step, we
apply the add operation to the line segments from $\OPT$ of level 0 that
intersects $S''$ (or to be more precise, their intersection with $S''$), i.e.,
to the set $\L':=\{\ell\cap S''\mid\ell\in\OPT\wedge\ell\cap
S''\ne\emptyset\wedge\ell\,\mathrm{is\,of\,level\,}0\}$. Claim~\ref
{clm:few-intersecting} implies that $|\L'|\le \eps^{-3}$. In each obtained
subproblem we apply the trivial operation until it is no longer applicable. We
say that all these operations correspond to level 0.

\paragraph{Subsequent levels.}
Next, we do a sequence of operations that correspond to levels $j=1,2,3,...$.
Assume by induction that for some $j$ each leaf in the current tree $T$
corresponds to a subproblem $\DP(S,\L)$ such that $\ell\cap S\in\L$ for each
line segment $\ell\in\OPT$ of each level $j'<j$ for which $\ell\cap
S\ne\emptyset$. Consider one of these leaves and suppose that it corresponds to
a subproblem $\DP(S,\L)$. We apply the line operation for each vertical line
that corresponds to a (vertical) grid line of level $j$.

Consider a corresponding subproblem $\DP(S',\L)$. Suppose that there more than
$\eps^{-3}$ line segments \ari{(horizontal or vertical)} from $\OPT$ of level
$j$ that have an endpoint inside $S'$. Like above, we consider these endpoints
and we order them non-decreasingly by their $y$-coordinates. Let
$p_{1},p_{2},...,p_{k}$ be these points in this order. For each $k'\in\N$ with
$k'/\eps^3\le k$, we consider the point $p_{k'/\eps^3}$ and apply the line
operation for the horizontal line $\ell'$ that contains $p_{k'/\eps^3}$. If for
a resulting subproblem $\DP(S'',\L)$ there is a vertical line segment
$\ell\in\L$ of some level $j'<j-2$ with an endpoint $p$ inside $S''$, then we
apply the line operation for the horizontal line that contains $p$.

\begin{figure}
    \centering
    \begin{subfigure}[h]{0.45\textwidth}
        \centering
        \includegraphics[width=\textwidth,page=1]{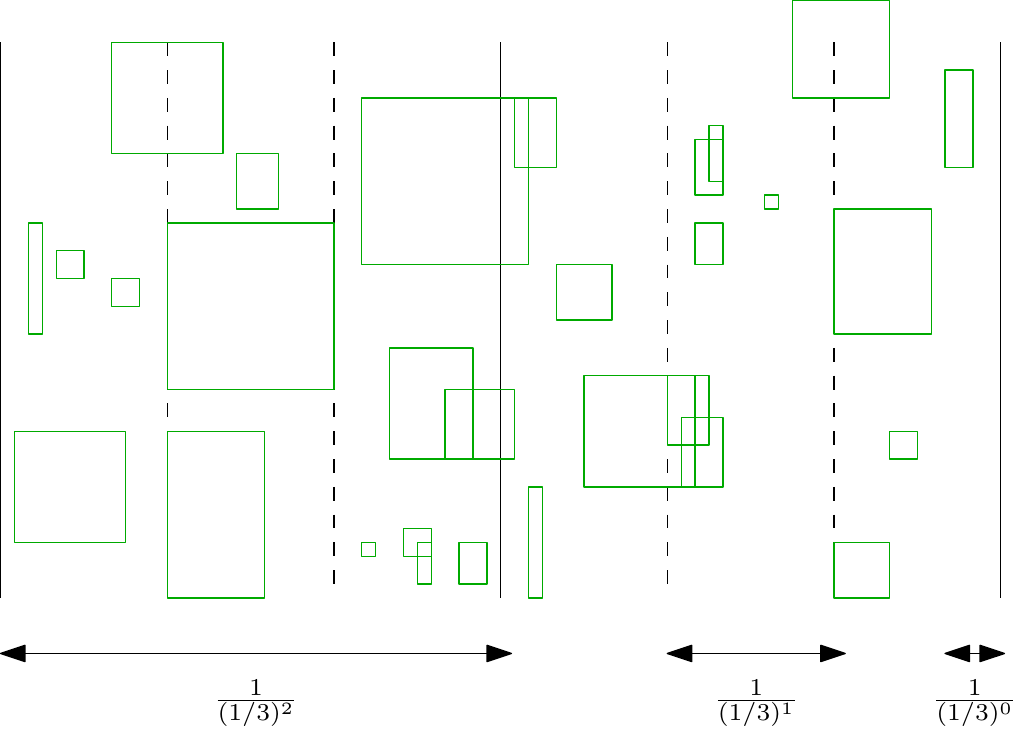}
        \caption{Input Instance}
    \end{subfigure}
    \hfill
    \begin{subfigure}[h]{0.45\textwidth}
        \centering
        \includegraphics[width=\textwidth,page=2]{03horzops}
        \caption{Lines in OPT and line operartions}
    \end{subfigure}
    \caption{Horizontal line operations}
    \label{fig:horzops}
\end{figure}

\begin{clm}
    \label{clm:few-intersecting-1}Let $\DP(S',\L)$ be one of the subproblems
    after applying the operations above. There are at most $\eps^{-3}$ line
    segments $\L'$ (horizontal or vertical) from $\OPT$ of level $j$
    that have an endpoint inside $S'$.
\end{clm}
Consider a resulting subproblem $\DP(S'',\L)$. For each line segment
$\ell\in\OPT$ such that $\ell$ crosses $S''$, i.e., $S''\setminus\ell$ has two
connected components, we apply the line operation to the line that contains
$\ell$. We apply the trivial operation until it is no longer applicable. In
each subproblem $\DP(S''',\L)$ obtained after this step, we apply the add
operation to the line segments of level $j$ that have an endpoint in $S'''$,
i.e., to the set $\L':=\{\ell\cap S'''\mid\ell\in\OPT\wedge\ell\cap
S'''\ne\emptyset\wedge\ell\,\mathrm{is\,of\,level\,}j\}$.

As an example, look at Figure~\ref{fig:horzops}, with $\eps=1/3$. The solid
lines in it are of level 0, the dashed lines of level 1, and dotted lines of
level 2.  Also the black lines are vertical grid lines, the blue lines are
(well-aligned) lines in $\OPT$ and the red lines are lines along which
horizontal line operations are applied. It can be seen from this example that a
segment of level $j$, by virtue of it being well-aligned, will get removed by a
trivial operation of level less than $j+3$.

\subsection{Analysis of DP-decision tree}

We want to prove that the resulting tree $T$ is indeed a DP-decision-tree
corresponding to a solution of cost at most $(1+\eps)\OPT$. To this end,
first we need to show that whenever we apply the add operation to a subproblem
$\DP(S,\L)$ for a set $\L'$ then $|\L|+|\L'|\le3\eps^{-3}$.  The {\em key insight}
for this is that if we added a line segment $\ell\in\OPT$ of some level $j$,
then it will not be included in the respective set $\L$ of later subproblems of
level $j+3$ or higher since $\ell$ is well-aligned. More precisely, if $\ell$
is horizontal then its $x$-coordinates are aligned with the grid lines of level
$j+3$. Hence, if $\ell$ or a part of $\ell$ is contained in a set $\L$ of some
subproblem $\DP(S,\L)$ for some level $j+3$, then we applied the trivial
operation to $\ell$ and thus $\ell$ ``disappeared'' from $\L$.
If $\ell$ is vertical and it appears in a $\DP(S,\L)$ for some level $j+3$ then
we applied the line operation to the horizontal lines that contain the two
endpoints of $\ell$. Afterwards, we applied the trivial operation to $\ell$
until $\ell$ ``disappeared'' from $\L$.

In particular, for each subproblem $\DP(S,\L)$ constructed by operations of
level $j$, the set $\L$ can contain line segments of levels $j-2,j-1$, and $j$;
but no line segments of a level $j'$ with $j'<j-2$. We make this and some other
technicalities formal in the proof of the next lemma.

\begin{lem}
    The constructed tree $T$ is a DP-decision-tree.
\end{lem}
\begin{proof}
    We first notice that any reduction in the tree corresponding to a trivial
    or line operation on $\DP(S, \L)$ can only lead to subproblems of the form
    $\DP(S', \L')$ with $\L'\subseteq\L$.  Hence, such reductions cannot lead
    to a contradiction of any DP-decision tree property.

    We now consider the add operations to ensure that any newly created node,
    by applying the add operation on set $\L':=\{\ell\cap S' \mid \ell\in\OPT
    \wedge\ell\cap S'\ne\emptyset \wedge \ell\mathrm{\,is\,of\,level\, }j\}$
    to subproblem $\DP(S,\L)$, maintains the property that
    $|\L|\cup|\L'|\le3\eps^{-3}$. We do this by induction,
    \begin{description}
        \item[Base case:] For the sequence of operations of level $j<3$, we
            know by Claim~\ref{clm:few-intersecting} and
            Claim~\ref{clm:few-intersecting-1} that the `added' set $\L'$
            always satisfied $|\L'|\le\eps^{-3}$. Since these are the only
            added sets in the first 3 levels, we can be sure that
            $|\L|+|\L'|\le3\eps^{-3}$.
        \item[Hypothesis:] We assume that all add operations in the sequence of
            operations corresponding to level $j-1$, satisfied the property
            that $|\L|+|\L'|\le3\eps^{-3}$.
        \item[Induction:] Consider an add operation for segments in $\OPT$ of
            level $j$.  Any such operation was preceded by a series of line
            operations along grid lines of level $j$ and all viable trivial
            operations. The grid lines of level $j$ are separated by
            $\eps^{j-2}$ units, and hence any horizontal segment from $\OPT$
            in $\L$ of level $j'<j-2$ (which have length $\ge\eps^{j-2}$
            and being well aligned, completely cut across any cell they
            intersect) can have the trivial operation applied to them, and be
            removed from $\L$. Similarly, for vertical segments in $\OPT$ of level
            $j'<j-2$, we explicitly apply the line operation to the horizontal
            lines that contain its end points, and hence again these also get
            removed from $\L$ when we apply all possible trivial operations.
            Thus we are left only with segments added in operations of level
            $j-1$ and  $j-2$ in  $\L$.  But we know from
            Claim~\ref{clm:few-intersecting-1} that at each level we add at
            most  $\eps^{-3}$ segments. Thus application of an add
            operation at this stage ensures that $|\L|+|\L'|\le3\eps^{-3}$.
        \qedhere
    \end{description}
\end{proof}
We want to show that the cost of the solution corresponding to $T$ is at most
$(1+O(\eps))\OPT$. In fact, depending on the offset $a$ this might or might
not be true. However, we show that there is a choice for $a$ such that this is
true (in fact, we will show that for a random choice for $a$ the cost will be
at most $(1+O(\eps))\OPT$ in expectation). Intuitively, when we apply the
line operation to a vertical grid line $\ell$ of some level $j$ then the
incurred cost is at most $O(1)$ times the cost of the line segments from $\OPT$
of level $j$ or larger that stab at least one rectangle intersected by $\ell$.
A line segment $\ell'\in\OPT$ of level $j$ stabs such a rectangle only if
$\ell'$ is intersected by $\ell$ (if $\ell'$ is horizontal) or the
$x$-coordinate of $\ell'$ is close to $\ell$ (if $\ell'$ is vertical). Here we
use that $h_{i}\ge w_{i}$ for each rectangle $R_{i}\in\OPT$.

Thus, we want to bound the total cost over all levels $j$ of the line segments
from $\OPT$ that are in level $j$ and that are intersected or close to grid
lines of level $j$ or smaller. We will show that if we choose $a$ randomly then
the total cost of such grid lines is at most $\eps\cdot\OPT$ in
expectation. Hence, by using the constant approximation algorithm from
\cite{ChanD0SW18} in expectation the total cost due to all line operations for
vertical line segments is at most $O(\eps)\cdot\OPT$.

When we apply the line operation for a horizontal line, then the cost of
stabbing the corresponding rectangles is at most the width of the rectangle $S$
of the current subproblem $\DP(S,\L)$. We will charge this cost to the line
segments of $\OPT$ inside $S$ of the current level or higher levels. We will
argue that we can charge each such line operation to line segments from $\OPT$
whose total width is at least $1/\eps$ times the width of $S$. This costs
another $O(\eps)\cdot\OPT$ in total due to all applications of line
operations for horizontal line segments.

The add operation yields a cost of exactly $\OPT$ and the trivial operation
does not cost anything. This yields a total cost of $(1+O(\eps))\OPT$.

Formally, we prove this in the following lemma.

\begin{lem} \label{lem:errorbound}
    There is a choice for the offset $a$ such that the solution
    $\overline{\SOL}([0,n]\times[0,4n^2],\emptyset)$ in $T$ has a cost of at most
    $(1+O(\eps))\OPT$.
\end{lem}
\begin{proof}
    In the tree as defined above, the add operations are only applied on
    segments from $\OPT$, and hence the cost across all such add operations is
    at most $c(\OPT)$.  Similarly, the trivial operations are applied on
    segments which were `added' before, and hence their cost is also already
    accounted for. So we are left with analyzing the cost of stabbing the
    rectangles which are intersected by the lines along which we apply the
    line operations. We claim that for a random offset $a$, this cost is
    $O(\eps\cdot\OPT)$, which would give us the required result.

    Let us first consider any line operation of level $j$ that is applied to a
    horizontal line $\ell$. This operation would create 2 cells of width at
    most $\eps^{j-2}$, one of which either contains $\eps^{-3}$ endpoints of
    segments \ari{(horizontal or vertical)} from $\OPT$ of level $j$; or
    contains \ari{at least one} vertical segment from $\OPT$ of level $j'<j-2$,
    i.e., the cost of the segments from $\OPT$ with at least one endpoint in
    this cell is at least $\eps^{-3}\cdot\eps^j = \eps ^{j-3}$.  Since a
    segment of width $\eps^{j-2}$ (width of cell) is sufficient to stab all
    rectangles stabbed by $\ell$, we see that this horizontal line takes only
    $\eps$ times the cost of the segments in $\OPT$ with at least one endpoint
    in the cell. We charge the cost of this horizontal segment to these
    corresponding endpoints. Since each such segment in $\OPT$ of level $j$ can
    be charged at most twice, by summing over all horizontal line operations
    over all levels we get that the cost of such line operations is at most
    $2\eps\cdot\OPT$.

    Now, let us consider the line operations applied to vertical grid lines.
    We wish to bound the cost of stabbing all the rectangles intersected or
    \textit{close to} grid lines (will be formally defined shortly), over all
    levels $j$. This as mentioned above can also be stated as bounding the
    cost, over all levels $j$, of line segments in level $j$ of $\OPT$ (call
    this set $\OPT_j$) intersected or close to grid lines of level $j$ or
    smaller. For a horizontal segment  $\ell\in\OPT_j$, let $I_\ell$ be the
    indicator variable representing the event that a grid line of level $j$ or
    smaller intersects $\ell$ ($I_\ell=0$ for vertical segments). Since
    $|\ell|\le\eps^{j-1}$, if we take a random offset $a$, we can say that,
    \[ \Exp[I_\ell] \le \frac{\eps^{j-1}}{\eps^{j-2}} = \eps.\]
    For a vertical segment $\ell\in\OPT_j$, let $J_\ell$ be the indicator
    variable representing the event that a grid line of level $j$ or smaller
    intersect the rectangle stabbed by $\ell$ ($J_\ell=0$ for horizontal
    segments). Since for $j>0$, $|\ell|\le\eps^{j-1}$, we know that for the
    rectangle stabbed by $\ell$, the dimensions satisfy $w_i\le h_i\le
    \eps^{j-1}$.  This means that to stab such a rectangle, $\ell$ has to
    lie \textit{close to}, i.e., within $\pm\eps^{j-1}$ of the vertical
    grid line. So for a random offset $a$ and level $j>0$ we can say that:
    \[ \Exp[J_\ell] \le \frac{2\eps^{j-1}}{\eps^{j-2}} = 2\eps.\]
    For level 0, we note that even though the vertical segments can be very
    long, the maximum width of a rectangle is at most 1. So $\ell$ has to lie
    within $\pm1$ of the grid line, giving us:
    \[ \Exp[J_\ell] \le \frac{2}{\eps^{-2}} = 2\eps^2\le 2\eps.\]
    With the expectations computed above, we can upper bound the expected cost
    of segments in $\OPT$ intersected by vertical line operations as:
    \begin{align*}
       \Exp\left[\sum_j\sum_{\ell\in\OPT_j} (I_\ell+J_\ell)\cdot|\ell|\right]
        &= \sum_j\sum_{\ell\in\OPT_j}
        \Exp\left[(I_\ell+J_\ell)\cdot|\ell|\right] \\
        &= \sum_j\sum_{\ell\in\OPT_j} |\ell|\cdot(\Exp[I_\ell]+\Exp[J_\ell]) \\
        &\le \sum_j\sum_{\ell\in\OPT_j} |\ell|\cdot (\eps + 2\eps) \\
        &=  3\eps\cdot\OPT
    \end{align*}
    Now, by using the $\alpha$-approximation algorithm for stabbing from
    \cite{ChanD0SW18}, where $\alpha$ is a constant, the solution returned by
    our algorithm takes an additional cost of $3 \alpha \cdot\eps\cdot\OPT$.
\end{proof}

Now we prove our main theorem.

\begin{thm}
    \label{thm:DP}There is a $(1+\eps)$-approximation algorithm
    for the general rectangle stabbing problem with a running time of
    $(n/\eps)^{O(1/\eps^3)}$, assuming that $h_{i}\ge w_{i}$ for each
    rectangle $R_{i}\in\R$.
\end{thm}
\begin{proof}
We gave a DP algorithm in Section~\ref{sec:DPdesc} which was shown to have the
required running time in Lemma~\ref{lem:runtime}. Further in Lemma~\ref
{lem:errorbound} we showed the correctness and that the solution computed by
the DP is actually a $1+O(\eps)$ approximation of the solution.
\end{proof}
Theorem~\ref{thm:DP} has some direct implications. First, it yields
a PTAS for the general square stabbing problem.
\begin{cor}
    There is a PTAS for the general square stabbing problem.
\end{cor}
Also, it yields a $(2+\eps)$-approximation algorithm for the general rectangle
stabbing problem for arbitrary rectangles: we can simply split the input into
rectangles $R_{i}$ for which $h_{i}\ge w_{i}$ holds, and those for which
$h_{i}<w_{i}$ holds, and output the union of these two solutions.
\begin{cor}
    There is a $(2+\eps)$-approximation algorithm for the general rectangle
    stabbing problem with a running time of $(n/\eps)^{O(1/\eps^3)}$.
\end{cor}
Finally, it yields a PTAS for the horizontal rectangle stabbing problem: we can
take the input of that problem and stretch all input rectangles
\ari{vertically} such that it is always very costly to stab any rectangle
vertically (so in particular our $(1+\eps)$-approximate solution would never do
this). Then we apply the algorithm due to Theorem~\ref{thm:DP}.
\begin{cor} \label{cor:horzstabbing}
    There is a $(1+\eps)$-approximation algorithm for the horizontal
    rectangle stabbing problem with a running time of
    $(n/\eps)^{O(1/\eps^3)}$.
\end{cor}

\section{$\delta$-large rectangles}
\label{sec:deltalarge}
We now consider the case of $\delta$-large rectangles for some given constant
$\delta$, i.e., where \aw{for each input rectangle $R_i$ we assume that $w_i
\le 1$ and $h_i \le 1$ and additionally $w_i \ge \delta$ or $h_i \ge
\delta$.}\awr{please check} In other words, none of the rectangles are small
($<\delta$) in both dimensions.  For this case we again give a PTAS in which we
use our algorithm due to Theorem~\ref{thm:DP} as a subroutine.

First, by losing only a factor of $1+\eps$, we divide the instance into
independent subproblems which are disjoint
rectangular cells. For each cell $C_i$, we denote by $\OPT(C_i)$ the cells from
$\OPT$ that are contained in $C_i$ and our routine ensures that
$c(\OPT(C_i))\le O(1/\eps^3)$.  Then for each cell $C_i$, the number of
segments in $\OPT(C_i)$ with length longer than $\delta$ is bounded by
$O({1}/{\delta\eps^3})$. We guess them in polynomial time.  Now, the remaining
segments in $\OPT$ are all of length smaller than $\delta$, and hence they can
stab a rectangle along its shorter dimension.  Hence, we can divide the
remaining rectangles into two disjoint sets, one with $h_i\aw{\ge}w_i$ and the
other with $w_i\ge h_i$, and use Theorem~\ref{thm:DP} to get a $1+\eps$
approximation of the remaining problem. In the following, we describe our
algorithm in detail.

\subsection{Guessing long segments in the solution}
We first start by discretizing the input similar to Lemma~\ref{lem:preprocess}.
\begin{lem} \label{lem:preprocess2}
    By losing a factor
    $(1+O(\eps))$ in the approximation ratio, we can assume for each
    $R_{i}\in\R$ the following properties hold:
    \begin{enumerate}
        \item[(i)] $x_{1}^{(i)},x_{2}^{(i)}, y_{1}^{(i)}, y_{2}^{(i)}$
            are discretized and within $[0, n]$,
        \item[(ii)] each line segment in $\OPT$ has width of at
            most $1/\eps$.
    \end{enumerate}
\end{lem}
\begin{proof}
    We first prove Property $(i)$. Since each rectangle $R_i\in\R$ has width
    $w_i\le 1$, the instance either has some $x$-coordinate that is not covered
    by any rectangle -- in which case we can split the instance into smaller
    subproblems around this $x$-coordinate -- or the total width of the instance
    is less than $n$.  Hence all the $x$-coordinates of the rectangles can be
    assumed to be between 0 and $n$. Further, we extend the rectangles on both
    sides to make their $x$-coordinates align with the next nearest multiple of
    ${\eps}/{n}$ (discretization). Clearly this involves extension by
    at most ${2\eps}/{n}$ to the width of every rectangle, and at most
    a $2\eps$ addition to the cost of the solution. We can use the same
    argument to show that the heights are also discretized and between 0 and
    $n$, since the heights of all rectangles are also less than 1. The height
    discretization similarly adds another $2\eps$ to the cost of the
    solution.

    Now we look at Property $(ii)$. Consider any horizontal (resp.~vertical)
    segment $\ell\in\OPT$, that is longer than $1/\eps$. From the left
    (resp.~bottom) end point, we divide the segment into consecutive smaller
    segments of length $1/\eps-2$ each, with one potential last piece being
    smaller than $1/\eps-2$. Now, for each smaller segment we extend it on
    both sides in such a way that it completely stabs the rectangles that it
    intersects. Since the maximum width (resp.~height) of a rectangle is 1, we
    extend each such segment by at most 2 units.

    In the worst case the highest possible fractional increase due to the above
    division happens when we have a segment of length $1 /\eps+\delta$ for
    very small $\delta>0$. But even in this case the fractional increase in the
    length of segments in $\OPT$ can be bounded by,
    \[\frac{1/\eps+\delta+4}{1/\eps+\delta}
    =1+4\cdot\frac{1}{1/\eps+\delta} \le 1+4\eps.  \qedhere\]
\end{proof}

\aw{
Then we continue by constructing vertical grid
lines at each $x$-coordinate $a+k\cdot\eps^{-2}$ with $k\in\N$ and
a random offset $a\in\N_0$. By removing the rectangles that intersect with these
grid lines, we divide the input instance into strips of width $1/\eps^2$
each. Next, we show that the removed rectangles can be stabbed at very low total cost.}

\begin{lem}
    Consider the input rectangles that are intersected by vertical grid lines.
    They can all be stabbed by a set of segments of total expected cost
    $O(\eps)\cdot\OPT$.
\end{lem}\awr{Or maybe use horizontal grid lines with the same spacing, and then arguing that the cost inside each cell is at most $O_{\delta,\eps}(1)$?}
\begin{proof}
    Similar to Lemma~\ref{lem:errorbound}, we wish to bound the cost of
    stabbing all the rectangles intersected or \textit{close to} grid lines,
    which can also be stated as bounding the cost, of line segments of $\OPT$
    intersected or close to the grid lines. For a horizontal segment
    $\ell\in\OPT$, let $I_\ell$ be the indicator variable representing the
    event that a grid line intersects $\ell$ ($I_\ell=0$ for vertical
    segments). Since $|\ell|\le1/\eps$, if we take a random offset $a$, we can
    say that,
    \[ \Exp[I_\ell] \le \frac{\eps^{-1}}{\eps^{-2}} = \eps.\]
    For a vertical segment $\ell\in\OPT$, let $J_\ell$ be the indicator
    variable representing the event that a grid line intersects the rectangle
    stabbed by $\ell$ ($J_\ell=0$ for horizontal segments). Since the width of
    all rectangles is less than 1, to stab a rectangle, $\ell$ has to lie
    \textit{close to}, i.e., within $\pm1$ of the vertical grid line. So for a
    random offset $a$ we can say that,
    \[ \Exp[J_\ell] \le \frac{2}{\eps^{-2}} = 2\eps^2\le 2\eps.\]

    With the expectations computed above, we can upper bound the expected cost
    of segments in $\OPT$ intersected by vertical line operations as,
    \begin{align*}
        \Exp\left[\sum_{\ell\in\OPT} (I_\ell+J_\ell)\cdot|\ell|\right]
         &= \sum_{\ell\in\OPT}
         \Exp\left[(I_\ell+J_\ell)\cdot|\ell|\right] \\
         &= \sum_{\ell\in\OPT} |\ell|\cdot(\Exp[I_\ell]+\Exp[J_\ell]) \\
         &\le \sum_{\ell\in\OPT} |\ell|\cdot (\eps + 2\eps) \\
         &=  3\eps\cdot\OPT
    \end{align*}
    Now, by using the $\alpha$-approximate algorithm for \orthstab~from
    \cite{ChanD0SW18}, where $\alpha$ is a constant, we can find a set of
    segments that stabs the set of rectangles intersected by the vertical grid
    lines with an additional cost of $3 \alpha \cdot\eps\cdot\OPT$.
\end{proof}

Now we divide these vertical strips into smaller rectangular cells such that
the optimal solution for each cell is $O(1/\eps^3)$.

\begin{lem}
    \aw{In time $n^{O(1)}$ we can compute a set of horizontal line segments of
    total cost $O(\eps) \cdot\OPT$ that, together with the vertical grid lines,
    divides}
    the input instance into rectangular cells containing independent
    subproblems, each with an optimal solution of cost $O(1/\eps^3)$.
\end{lem}
\begin{proof}
    We will consider each vertical strip of width at most $1/\eps^2$ separately
    as they constitute independent subproblems.  Consider one such strip. We
    sweep a horizontal line from bottom to top, and at each discretized
    $y$-coordinate compute an  $\alpha$-approximate  solution (using
    $\alpha$-approximation algorithm of \cite{ChanD0SW18} for \orthstab) for
    the subproblem completely contained in the cell below the line. At the
    smallest $y$-coordinate $y_0$, at which the cost of the solution exceeds
    $1/\eps^3$, we take the line $[0,1/\eps^2]\times y_0$ as a part of our
    solution. Note that by moving the line up by another $\eps/n$, the cost of
    the solution can increase by at most $1/\eps$, therefore the cost of
    segments from $\OPT$ in the cell thus created below the line is between
    $1/(\alpha \eps^3)$ and $1/\eps^3+1/\eps$, which is $\Omega(1/\eps^3)$.

    Thus every time we add a horizontal segment as described above, we add a
    segment of length $1/\eps^2$ to the solution, and create a subproblem which
    has an optimal solution of cost $\Omega(1/\eps^3)$. Hence the cost of the
    added segment is at most $O(\eps^{-2}/\eps^{-3})=O(\eps)$ times the cost of
    the solution to the subproblem that it creates. Since an optimal solution
    to  all such cells can be at most the cost of the original instance, the
    cost of the horizontal segments that we have added is at most $O(\eps)
    \cdot\OPT$.
\end{proof}

Now, in the solution to a subproblem contained in a cell as described above,
there can be at most $O(1/\delta\eps^3)$ \textit{long segments}, i.e.,
segments of length at least $\delta$. Since there are only a polynomial number
of discretized candidate segments we can guess the set of long segments in the
solution in polynomial time (the exact analysis is shown later).

\subsection{Computing short segments in the solution}
For each guess made for the set of long segments in the solution, we compute a
$(1+\eps)$-approximation for the \textit{short segments}, i.e., segments of
length at most $\delta$ in the solution using Theorem~\ref{thm:DP}.

Note that any rectangle with both height and width at least $\delta$ will have
to be stabbed by a segment longer than $\delta$. However, since all segments in
the solution of length at least $\delta$ have already been guessed, \aw{each}
remaining segment can stab a $\delta$-large rectangle only along its shorter
dimension. So the rectangles in the instance that have not yet been stabbed can
be partitioned into {\em two} independent instances, one in which has
rectangles with width less than $\delta$ (will be stabbed horizontally) and one
which has rectangles with height less than $\delta$ (will be stabbed
vertically). We apply Theorem~\ref{thm:DP} on each of these
independent instances and combine the solutions to get a $(1+\eps)$-approximate
solution for the remaining problem.

\begin{thm}
     For \orthstab~with $\delta$-large rectangles, there is a
     $(1+\eps)$-approximation algorithm with a running time of
     $(n/\eps)^{O(1/\delta\eps^3)}$.
\end{thm}
\begin{proof}
    First, we analyze the running time.  In the first phase of the algorithm,
    we divide the given instance into ${n}/({1/\eps^2})=n\eps^2$ strips, and
    run the constant factor $n^{O(1)}$ approximation algorithm from
    \cite{ChanD0SW18} for each discretized $y$-coordinate. This requires a
    running time of $n\eps^2 \times {n}/({\eps/n}) \times n^{O(1)} = \eps\cdot
    n^{O(1)}$.  Now in a strip there are at most $({n}/{(\eps/n)}) \times
    ({(1/\eps^2)}/{(\eps/n)})=n^3/\eps^4$ possible discretized end points for
    segments, which means there are less than $\binom{n^3/\eps^4}{2}$, i.e.,
    $O(n^4/\eps^8)$ possible candidate segments in a strip. Hence, we can
    enumerate over all of their subsets of size $O(1/\delta\eps^3)$ in
    $(n/\eps)^{O(1/\delta\eps^3)}$ time.  Following this we have two
    invocations of Theorem~\ref{thm:DP} to solve the independent subproblems
    which would together take $(n/\eps)^{O(1/\eps^3)}$. This brings the overall
    running time of the algorithm to be $(n/\eps)^{O(1/\delta\eps^3)}$.

    The correctness of the algorithm follows from the fact that we can
    correctly guess all the $\delta$-large segments in $\OPT$, and once we have
    them the remaining problem gets divided into independent subproblems, for
    each of which we can compute a $(1+\eps)$-approximation. So the
    approximation factor for the overall algorithm is also $(1+\eps)$.
\end{proof}

\section{Generalized minimum Manhattan network} \label{sec:gmmn}
In the generalized minimum Manhattan network problem (GMMN), we are given a set
$R$ of $n$ unordered terminal pairs, and the goal is to find a minimum-length
rectilinear network such that every pair in $R$ is M-connected, that is,
connected by an Manhattan-path, which is a path consisting of axis parallel
segments whose total length is the Manhattan distance of the pair under
consideration.

For the special case of 2D-GMNN, where all terminals lie on a 2 dimensional
plane, \cite{das2018approximating} gave a $O(\log n)$-approximate solution
using a $(6+\eps)$-approximate solution for $x$-separated 2D-GMMN instances. An
instance of GMMN is called $x$-separated if there is a vertical line that
intersects the rectangle formed by any pair of terminals as the diagonally
opposite corners of the rectangle. Their algorithm solves a $x$-separated
instance by giving a solution of the form $N=A_{\text{up}}\cup A_{\text
{down}}\cup S$, where $S$ is a set of horizontal segments stabbing all the
rectangles in the instance and $A_{\text{up}}\cup A_{\text{down}}$ is a set of
segments of cost bounded by $(2+2\eps)\cdot(\OPT+\OPT_{\text{ver}})$.  Using a
4-approximate algorithm for finding the set of segments $S$, horizontally
stabbing the rectangles, a bound of $(6+\eps)\cdot\OPT$ is obtained for the
cost of segments in $N$.  But instead of the 4-approximate algorithm we can
instead use the PTAS from Corollary~\ref{cor:horzstabbing} to obtain a
better bound of $(4+\eps)\cdot\OPT$ on the cost as follows,
\begin{align*}
    \aw{c(N)} &= \aw{c(A_\text{up})} + \aw{c(A_\text{down})} + \aw{c(S)} \\
    &\le (2+2\eps)\cdot(\OPT+\OPT_{\text{ver}}) + (1+\eps)\cdot\OPT_{\text{hor}} \\
    &\le (2+2\eps)\cdot\OPT + (2+2\eps)\cdot(\OPT_{\text{ver}}+\OPT_{\text{hor}}) \\
    &= (4+\eps')\cdot\OPT.
\end{align*}
Note that in the above equations $\OPT_\text{ver}$ and $\OPT_\text{hor}$ refer to the
cost of the vertical and horizontal segments respectively in an optimal solution.

\bibliographystyle{alpha}
\bibliography{references}
\end{document}

%% file: introduction.tex
\section{Introduction}
\label{sec:intro} Rectangle stabbing problems are natural geometric
optimization problems.  Here, we are given a set of $n$ axis-parallel
rectangles $\R$ in the two-dimensional plane. For each rectangle $R_{i}\in\R$, 
we are given points
$(x_{1}^{(i)},y_{1}^{(i)}),(x_{2}^{(i)},y_{2}^{(i)})\in\mathbb{R}^{2}$ that
denote its bottom-left and top-right corners, respectively. Also, we denote its
{\em width} and {\em height} by $w_{i}:=x_{2}^{(i)}-x_{1}^{(i)}$ and
$h_{i}:=y_{2}^{(i)}-y_{1}^{(i)}$, respectively.  
\aw{Our goal is to \ari{compute} a set of line segments $\L$ that
stab all input rectangles.}
We call a rectangle {\em stabbed}
if a segment $\ell \in \L$ intersects both of its horizontal or both of its
vertical edges.
We study several variants. In the \textit{horizontal rectangle stabbing
problem} (\horstab) \aw{we want to find a} set of horizontal segments of minimum
total length such that each rectangle is stabbed. \aw{The} \textit{general rectangle
stabbing} (\orthstab) problem generalizes \horstab~and  involves finding a set
of axis parallel segments of minimum total length such that each rectangle in
$\R$ is stabbed. \aw{The} \textit{general square stabbing} (\sqstab) problem is a
special case of \orthstab~where all rectangles in the input instance are
squares.  These problems have applications in bandwidth allocation, message
scheduling with time-windows on a direct path, and geometric network design
\cite{ChanD0SW18, becchetti2009latency, das2018approximating}.

Note that
\horstab~and \orthstab~\aw{are special cases} of weighted geometric set cover, 
where the rectangles correspond to elements and potential line
segments correspond to sets, and the weight of a set equals the length of the corresponding segment.  A set
contains an element if the corresponding line segment stabs the corresponding
rectangle.  This already implies an $O(\log n)$-approximation algorithm
\cite{chvatal1979greedy} for \orthstab~and \horstab.  

Chan, van Dijk, Fleszar, Spoerhase, and Wolff~\cite{ChanD0SW18} initiated the study of \horstab. They proved
\horstab~to be NP-hard via a reduction from planar vertex cover.  
Also, they presented a constant\footnote{The constant is not
explicitly stated, and it depends on a not explicitly stated constant in \cite{chan2012weighted}.} factor approximation using {\em decomposition techniques} and the
{\em quasi-uniform sampling} method \cite{Varadarajan10} for weighted geometric
set cover.  In particular, they showed that \horstab~instances can be decomposed
into two disjoint {\em laminar} set cover instances of small shallow cell
complexity for which the {\em quasi-uniform sampling} yields a constant
approximation using techniques from \cite{ChanGKS12}.


Recently, Eisenbrand, Gallato, Svensson, and Venzin~\cite{QPTAS} presented a quasi-polynomial time
approximation scheme (QPTAS) for \horstab.  This shows that \horstab~is not APX-hard
unless $\text{NP} \subseteq \text{DTIME}(2^{\text{poly} \log n})$.  
\ari{The QPTAS relies on {\em the shifting technique} by Hochbaum and Maass
\cite{hochbaum1985approximation}, applied to a grid, consisting of
randomly shifted vertical grid lines that are equally spaced. With this approach, the plane is
partitioned into narrow disjoint vertical strips which they 
then process further. Then, this routine is applied recursively.}
%
They also gave a polynomial time dynamic programming based exact
algorithm for \horstab~for laminar instances (in which the projections of
the rectangles to the $x$-axis yield a laminar family of intervals). Then they
provided a simple polynomial-time 8-approximation algorithm by \aw{reducing
any given instance to a laminar instance. It remains open whether there is a PTAS for the problem.}
%

\begin{figure}
    \centering
    \begin{subfigure}[h]{0.4\textwidth}
        \centering
        \includegraphics[width=0.6\textwidth,page=1]{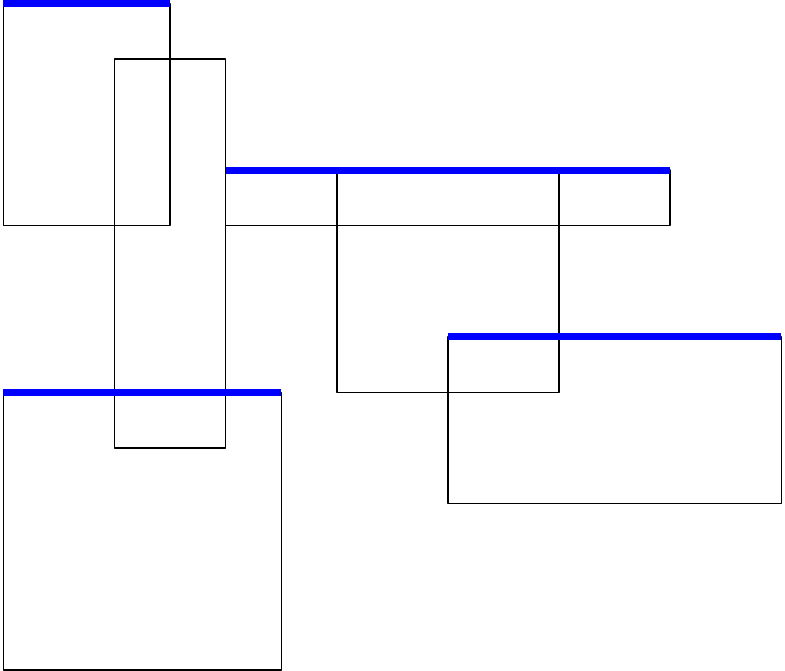}
        \caption{Horizontal Rectangle Stabbing}
    \end{subfigure}
    \begin{subfigure}[h]{0.4\textwidth}
        \centering
        \includegraphics[width=0.6\textwidth,page=2]{02genstab}
        \caption{General Rectangle Stabbing}
    \end{subfigure}
    \caption{A solution for an instance of \horstab\ and \orthstab.}
    \label{fig:genstab}
\end{figure}

\subsection{Our results}
\aw{In this paper, we give a PTAS for \horstab~and thus resolve this open question.
Also, we extend our techniques to \orthstab~for which we present a polynomial time $(2+\eps)$-approximation 
and PTASs for several special cases: when all input rectangles are squares, more generally when for each rectangle its
width is at most its height, and finally for $\delta$-large rectangles, i.e., when each rectangle has one edge whose
length is within $[\delta,1]$ and 1 is the maximum length of each edge of any input rectangle (in each dimension).
}

\aw{Our algorithm for \horstab~is in fact quite easy to state: it is a dynamic program (DP) that recursively
subdivides the plane into smaller and smaller rectangular regions. In the process, it guesses line segments from~OPT.}
However, its analysis is intricate. We show that there is a sequence of recursive decompositions that yields
a solution whose overall cost is $(1+\eps)\OPT$.
Instead of using a set of equally spaced grid lines as in \cite{QPTAS}, we use a \emph{hierarchical} grid with several levels for the decomposition.
In each level of our decomposition, we subdivide the given 
rectangular region
into strips of narrow width and guess $O_\eps(1)$ line segments from OPT inside them
which correspond to the current level. One crucial ingredient is that we slightly extend the segments, such that the guessed horizontal line segments are aligned with
our grid. The key consequence is that it will no longer be necessary to remember these line segments once we have advances three levels further in the
decomposition. Also, for the guessed vertical line segments of the current level we introduce additional (very short) horizontal line segments, such that
we do not need to remember them either, once we advanced three levels more in the decomposition. Therefore, the DP needs to remember 
previously guessed line segments from only the last three previous levels 
and afterwards these line segments {\em vanish}. This allows us to bound the number of arising subproblems (and hence of the DP-cells) 
by a polynomial. 

Our techniques easily generalize to a PTAS for \sqstab~and to a PTAS for 
\orthstab~if for each input rectangle its
width is at most its height, whereas the QPTAS in \cite{QPTAS} worked only for \horstab. 
We use the latter PTAS as a subroutine in order to obtain a polynomial time $(2+\eps)$-approximation for \orthstab~-- which improves
on the \ari{result}  in \cite{ChanD0SW18}. 

\aw{Then, we extend our techniques above to the setting of $\delta$-large rectangles of \orthstab. 
\ari{This is an important subclass of rectangles and are well-studied for other geometric problems \cite{AdamaszekW13, KP20}}. 
To this end, we first reduce the problem
to the setting in which all input rectangles are contained in a rectangular box that admits a solution of cost $O(1/\eps^3)$.
Then we guess the relatively long
line segments in OPT in polynomial time. The key argument is that then the remaining problem splits into two independent subproblems, one for the horizontal 
and one for the vertical line segments in OPT. For each of those, we then apply our PTAS for \horstab~which then yields a PTAS for 
\orthstab~if all rectangles are $\delta$-large.
}
%
%
%
%
%
%
%

Finally, \ari{our PTAS for \horstab~implies improved approximation ratios for the 
{\sc Generalized  Minimum Manhattan Network} (GMMN)
and  $x$-separated 2D-GMMN problems, of  $(4+\eps)\log n$ and  $4+\eps$, respectively, 
by improving certain subroutines
of the algorithm in \cite{das2018approximating}.}
%
%
%

\subsection{Further related work}
Finke et al.~\cite{finke2008batch}  gave a polynomial time exact algorithm for
a special case of \horstab~where all input rectangles have their left edges
lying on the $y$-axis. Das et al.~\cite{das2018approximating} studied a related
variant in the context of the {\sc Generalized 
Minimum Manhattan Network} (GMMN) problem. 
In GMMN, we are given a set of $n$ terminal-pairs and the goal 
is to find a minimum-length rectilinear network such that each pair is connected by a 
{\em Manhattan path}. 
They obtained a $4$-approximation for a variant of \horstab~where all rectangles intersect a vertical line.
Then they used it to obtain a $(6+\eps)$-approximation algorithm for the $x$-separated 2D-GMMN problem,
 a special case of 2D-GMMN, and $(6+\eps)(\log
n)$-approximation for 2-D GMMN.  

Gaur et al.~\cite{gaur2002constant} studied \aw{the problem of}
stabbing rectangles by a {\em minimum number of axis-aligned lines} and gave an
LP-based  2-approximation algorithm. Kovaleva and Spieksma
\cite{kovaleva2006approximation} considered a weighted generalization of this
problem and gave an $O(1)$-approximation algorithm.

\aw{\horstab~and \orthstab~are related to}
geometric set cover which is a fundamental geometric optimization problem.
Br{\"{o}}nnimann and Goodrich \cite{BronnimannG95} in a seminal paper gave an
$O(d \log (d \OPT))$-approximation for unweighted geometric set cover where $d$
is the dual VC-dimension of the set system and $\OPT$ is the value of the
optimal solution.  Using {\em $\eps$-nets}, Aronov et
al.~\cite{aronov2010small}  gave an $O(\log \log OPT)$-approximation for
hitting set for axis-parallel rectangles.  Later, Varadarajan
\cite{Varadarajan10} developed  quasi-uniform sampling  and provided
sub-logarithmic  approximation for weighted set cover where sets are weighted
fat triangles or weighted disks. Chan et al.~\cite{ChanGKS12} generalized this
to any set system with low {\em shallow cell complexity}.  Bansal and Pruhs
\cite{bansal2014geometry} reduced several scheduling problems to a particular
geometric set cover problem for anchored rectangles and obtained
$O(1)$-approximation using quasi-uniform sampling.  Afterward, Chan and Grant
\cite{ChanG14} and Mustafa et al.~\cite{MustafaRR15} have settled the
APX-hardness statuses of all natural weighted geometric set cover problems.

Another related problem is maximum independent set of rectangles (MISR).
Adamaszek and Wiese \cite{AdamaszekW13} gave a QPTAS for MISR using 
{\em balanced} cuts with small complexity that intersect only a few
rectangles in the solution.  Recently, Mitchell \cite{mitchell21} obtained the
first polynomial time constant approximation algorithm for the problem,
followed by a $(2+\eps)$-approximation by  G{\'{a}}lvez et al.~\cite{Galvez22}.

Many other rectangle packing problems are also well-studied
\cite{ChristensenKPT17},  such as geometric knapsack \cite{GalvezGHI0W17,
Galvez00RW21, KMSW21}, geometric bin packing \cite{BansalK14, KhanS21,
KhanSS21}, strip packing \cite{Galvez0AJ0R20, DeppertJ0RT21, harren20145},
storage allocation problem \cite{MomkeW20, MomkeW15}, etc.

%% file: main.bbl
\newcommand{\etalchar}[1]{$^{#1}$}
\begin{thebibliography}{BMSV{\etalchar{+}}09}

\bibitem[AES10]{aronov2010small}
Boris Aronov, Esther Ezra, and Micha Sharir.
\newblock Small-size $\backslash$eps-nets for axis-parallel rectangles and
  boxes.
\newblock {\em SIAM Journal on Computing}, 39(7):3248--3282, 2010.

\bibitem[AW13]{AdamaszekW13}
Anna Adamaszek and Andreas Wiese.
\newblock Approximation schemes for maximum weight independent set of
  rectangles.
\newblock In {\em FOCS}, pages 400--409, 2013.

\bibitem[BG95]{BronnimannG95}
Herv{\'{e}} Br{\"{o}}nnimann and Michael~T. Goodrich.
\newblock Almost optimal set covers in finite vc-dimension.
\newblock {\em Discret. Comput. Geom.}, 14(4):463--479, 1995.

\bibitem[BK14]{BansalK14}
Nikhil Bansal and Arindam Khan.
\newblock Improved approximation algorithm for two-dimensional bin packing.
\newblock In {\em SODA}, pages 13--25, 2014.

\bibitem[BMSV{\etalchar{+}}09]{becchetti2009latency}
Luca Becchetti, Alberto Marchetti-Spaccamela, Andrea Vitaletti, Peter Korteweg,
  Martin Skutella, and Leen Stougie.
\newblock Latency-constrained aggregation in sensor networks.
\newblock {\em ACM Transactions on Algorithms (TALG)}, 6(1):1--20, 2009.

\bibitem[BP14]{bansal2014geometry}
Nikhil Bansal and Kirk Pruhs.
\newblock The geometry of scheduling.
\newblock {\em SIAM Journal on Computing}, 43(5):1684--1698, 2014.

\bibitem[CG14]{ChanG14}
Timothy~M. Chan and Elyot Grant.
\newblock Exact algorithms and apx-hardness results for geometric packing and
  covering problems.
\newblock {\em Comput. Geom.}, 47(2):112--124, 2014.

\bibitem[CGKS12a]{chan2012weighted}
Timothy~M Chan, Elyot Grant, Jochen K{\"o}nemann, and Malcolm Sharpe.
\newblock Weighted capacitated, priority, and geometric set cover via improved
  quasi-uniform sampling.
\newblock In {\em Proceedings of the twenty-third annual ACM-SIAM symposium on
  Discrete Algorithms}, pages 1576--1585. SIAM, 2012.

\bibitem[CGKS12b]{ChanGKS12}
Timothy~M. Chan, Elyot Grant, Jochen K{\"{o}}nemann, and Malcolm Sharpe.
\newblock Weighted capacitated, priority, and geometric set cover via improved
  quasi-uniform sampling.
\newblock In Yuval Rabani, editor, {\em SODA}, pages 1576--1585, 2012.

\bibitem[Chv79]{chvatal1979greedy}
Vasek Chvatal.
\newblock A greedy heuristic for the set-covering problem.
\newblock {\em Mathematics of operations research}, 4(3):233--235, 1979.

\bibitem[CKPT17]{ChristensenKPT17}
Henrik~I Christensen, Arindam Khan, Sebastian Pokutta, and Prasad Tetali.
\newblock Approximation and online algorithms for multidimensional bin packing:
  {A} survey.
\newblock {\em Computer Science Review}, 24:63--79, 2017.

\bibitem[CvDF{\etalchar{+}}18]{ChanD0SW18}
Timothy~M. Chan, Thomas~C. van Dijk, Krzysztof Fleszar, Joachim Spoerhase, and
  Alexander Wolff.
\newblock Stabbing rectangles by line segments - how decomposition reduces the
  shallow-cell complexity.
\newblock In Wen{-}Lian Hsu, Der{-}Tsai Lee, and Chung{-}Shou Liao, editors,
  {\em ISAAC}, volume 123, pages 61:1--61:13, 2018.

\bibitem[DFK{\etalchar{+}}18]{das2018approximating}
Aparna Das, Krzysztof Fleszar, Stephen Kobourov, Joachim Spoerhase, Sankar
  Veeramoni, and Alexander Wolff.
\newblock Approximating the generalized minimum manhattan network problem.
\newblock {\em Algorithmica}, 80(4):1170--1190, 2018.

\bibitem[DJK{\etalchar{+}}21]{DeppertJ0RT21}
Max~A Deppert, Klaus Jansen, Arindam Khan, Malin Rau, and Malte Tutas.
\newblock Peak demand minimization via sliced strip packing.
\newblock In {\em APPROX/RANDOM}, volume 207, pages 21:1--21:24, 2021.

\bibitem[EGSV21]{QPTAS}
Friedrich Eisenbrand, Martina Gallato, Ola Svensson, and Moritz Venzin.
\newblock A {QPTAS} for stabbing rectangles.
\newblock {\em CoRR}, abs/2107.06571, 2021.

\bibitem[FJQS08]{finke2008batch}
Gerd Finke, Vincent Jost, Maurice Queyranne, and Andr{\'a}s Seb{\H{o}}.
\newblock Batch processing with interval graph compatibilities between tasks.
\newblock {\em Discrete Applied Mathematics}, 156(5):556--568, 2008.

\bibitem[GGA{\etalchar{+}}20]{Galvez0AJ0R20}
Waldo G{\'{a}}lvez, Fabrizio Grandoni, Afrouz~Jabal Ameli, Klaus Jansen,
  Arindam Khan, and Malin Rau.
\newblock A tight (3/2+{\(\epsilon\)}) approximation for skewed strip packing.
\newblock In {\em APPROX/RANDOM}, volume 176, pages 44:1--44:18, 2020.

\bibitem[GGH{\etalchar{+}}17]{GalvezGHI0W17}
Waldo G{\'{a}}lvez, Fabrizio Grandoni, Sandy Heydrich, Salvatore Ingala,
  Arindam Khan, and Andreas Wiese.
\newblock Approximating geometric knapsack via l-packings.
\newblock In {\em FOCS}, pages 260--271, 2017.

\bibitem[GGK{\etalchar{+}}21]{Galvez00RW21}
Waldo G{\'{a}}lvez, Fabrizio Grandoni, Arindam Khan, Diego
  Ram{\'{\i}}rez{-}Romero, and Andreas Wiese.
\newblock Improved approximation algorithms for 2-dimensional knapsack: Packing
  into multiple l-shapes, spirals, and more.
\newblock In {\em SoCG}, volume 189, pages 39:1--39:17, 2021.

\bibitem[GIK02]{gaur2002constant}
Daya~Ram Gaur, Toshihide Ibaraki, and Ramesh Krishnamurti.
\newblock Constant ratio approximation algorithms for the rectangle stabbing
  problem and the rectilinear partitioning problem.
\newblock {\em Journal of Algorithms}, 43(1):138--152, 2002.

\bibitem[GKM{\etalchar{+}}21]{Galvez22}
Waldo G{\'{a}}lvez, Arindam Khan, Mathieu Mari, Tobias M{\"{o}}mke,
  Madhusudhan~Reddy Pittu, and Andreas Wiese.
\newblock A (2+$\epsilon$)-approximation algorithm for maximum independent set
  of rectangles.
\newblock {\em CoRR}, abs/2106.00623, 2021.

\bibitem[HJPVS14]{harren20145}
Rolf Harren, Klaus Jansen, Lars Pr{\"a}del, and Rob Van~Stee.
\newblock A (5/3+ $\varepsilon$)-approximation for strip packing.
\newblock {\em Computational Geometry}, 47(2):248--267, 2014.

\bibitem[HM85]{hochbaum1985approximation}
Dorit~S. Hochbaum and Wolfgang Maass.
\newblock Approximation schemes for covering and packing problems in image
  processing and vlsi.
\newblock {\em Journal of the ACM (JACM)}, 32(1):130--136, 1985.

\bibitem[KMSW21]{KMSW21}
Arindam Khan, Arnab Maiti, Amatya Sharma, and Andreas Wiese.
\newblock On guillotine separable packings for the two-dimensional geometric
  knapsack problem.
\newblock In {\em SoCG}, volume 189, pages 48:1--48:17, 2021.

\bibitem[KP20]{KP20}
Arindam Khan and Madhusudhan~Reddy Pittu.
\newblock On guillotine separability of squares and rectangles.
\newblock In {\em APPROX/RANDOM}, volume 176, pages 47:1--47:22, 2020.

\bibitem[KS06]{kovaleva2006approximation}
Sofia Kovaleva and Frits Spieksma.
\newblock Approximation algorithms for rectangle stabbing and interval stabbing
  problems.
\newblock {\em SIAM Journal on Discrete Mathematics}, 20(3):748--768, 2006.

\bibitem[KS21]{KhanS21}
Arindam Khan and Eklavya Sharma.
\newblock Tight approximation algorithms for geometric bin packing with skewed
  items.
\newblock In {\em APPROX/RANDOM}, volume 207, pages 22:1--22:23, 2021.

\bibitem[KSS21]{KhanSS21}
Arindam Khan, Eklavya Sharma, and K.V.N. Sreenivas.
\newblock Geometry meets vectors: Approximation algorithms for multidimensional
  packing.
\newblock {\em CoRR}, abs/2106.13951, 2021.

\bibitem[Mit21]{mitchell21}
Joseph S.~B. Mitchell.
\newblock Approximating maximum independent set for rectangles in the plane.
\newblock {\em CoRR}, abs/2101.00326, 2021.

\bibitem[MRR15]{MustafaRR15}
Nabil~H. Mustafa, Rajiv Raman, and Saurabh Ray.
\newblock Quasi-polynomial time approximation scheme for weighted geometric set
  cover on pseudodisks and halfspaces.
\newblock {\em {SIAM} J. Comput.}, 44(6):1650--1669, 2015.

\bibitem[MW15]{MomkeW15}
Tobias M{\"{o}}mke and Andreas Wiese.
\newblock A (2+$\varepsilon$)-approximation algorithm for the storage
  allocation problem.
\newblock In {\em ICALP}, volume 9134, pages 973--984, 2015.

\bibitem[MW20]{MomkeW20}
Tobias M{\"{o}}mke and Andreas Wiese.
\newblock Breaking the barrier of 2 for the storage allocation problem.
\newblock In {\em ICALP}, volume 168, pages 86:1--86:19, 2020.

\bibitem[Var10]{Varadarajan10}
Kasturi~R. Varadarajan.
\newblock Weighted geometric set cover via quasi-uniform sampling.
\newblock In {\em STOC}, pages 641--648, 2010.

\end{thebibliography}
